\documentclass[a4paper]{llncs}
\pdfoutput=1

\usepackage[utf8]{inputenc}
\usepackage[english]{babel}  
\usepackage[T1]{fontenc} 
\usepackage{IEEEtrantools}
\usepackage{epsfig}
\usepackage{amsmath, amssymb, amsbsy}
\usepackage{xspace}          
\usepackage[noend]{algpseudocode}
\usepackage{algorithm}
\usepackage{algorithmicx}
\usepackage{color}
\usepackage{cite}            
\usepackage{booktabs} 
\usepackage{verbatim} 
\usepackage{url}
\usepackage[final,tracking=true,kerning=true,spacing=true,factor=1100,stretch=10,shrink=20]{microtype}
\usepackage{enumitem}

\title{Further Generalisations of \\ Twisted Gabidulin Codes}
\author{Sven Puchinger\inst{1} \and Johan Rosenkilde n\'e Nielsen\inst{2} \and John Sheekey\inst{3}}
\institute{Institute of Communications Engineering, Ulm University, Ulm, Germany\\
\email{sven.puchinger@uni-ulm.de}
\and
Department of Applied Mathematics \& Computer Science, Technical University of Denmark, Lyngby, Denmark\\
\email{jsrn@jsrn.dk}
\and
School of Mathematics and Statistics, University College Dublin, Dublin, Ireland. 
\email{john.sheekey@ucd.ie}
}

\usepackage{varioref}        
\usepackage{fancyref}

\makeatletter
\def\mkfancyprefix#1#2{%
\expandafter\def\csname fancyref#1labelprefix\endcsname{#1}%
\begingroup\def\x{\endgroup\frefformat{plain}}%
    \expandafter\x\csname fancyref#1labelprefix\endcsname
    {\MakeLowercase{#2}\fancyrefdefaultspacing##1}%
\begingroup\def\x{\endgroup\Frefformat{plain}}%
    \expandafter\x\csname fancyref#1labelprefix\endcsname
    {#2\fancyrefdefaultspacing##1}%
\begingroup\def\x{\endgroup\frefformat{vario}}%
    \expandafter\x\csname fancyref#1labelprefix\endcsname
    {\MakeLowercase{#2}\fancyrefdefaultspacing##1##3}%
\begingroup\def\x{\endgroup\Frefformat{vario}}%
    \expandafter\x\csname fancyref#1labelprefix\endcsname
    {#2\fancyrefdefaultspacing##1##3}%
}
\makeatother
\fancyrefchangeprefix{\fancyrefeqlabelprefix}{eqn}
\mkfancyprefix{ssec}{Section}
\mkfancyprefix{tbl}{Table}
\mkfancyprefix{thm}{Theorem}
\mkfancyprefix{lem}{Lemma}
\mkfancyprefix{cor}{Corollary}
\mkfancyprefix{prop}{Proposition}
\mkfancyprefix{prob}{Problem}
\mkfancyprefix{alg}{Algorithm}
\mkfancyprefix{inv}{Invariant}
\mkfancyprefix{ex}{Example}
\mkfancyprefix{line}{Line}
\mkfancyprefix{def}{Definition}
\mkfancyprefix{itm}{Item}
\mkfancyprefix{app}{Appendix}
\mkfancyprefix{rem}{Remark}
\newcommand{\cref}[1]{\Fref{#1}}

\def\ve#1{{\mathchoice{\mbox{\boldmath$\displaystyle #1$}}%
              {\mbox{\boldmath$\textstyle #1$}}%
              {\mbox{\boldmath$\scriptstyle #1$}}%
              {\mbox{\boldmath$\scriptscriptstyle #1$}}}}

\definecolor{mygreen}{rgb}{0,0.7,0}

\newcommand{\NN}{\mathbb{N}}
\newcommand{\Fq}{\mathbb{F}_q}
\newcommand{\Fqm}{\mathbb{F}_{q^m}}
\newcommand{\Fqs}{\mathbb{F}_{q^s}}

\newcommand{\Code}{\mathcal{C}}

\newcommand{\A}{\ve A}

\newcommand{\B}{\ve B}
\newcommand{\C}{\mathcal{C}}

\newcommand{\x}{\ve x}

\DeclareMathOperator{\evOp}{ev}
\newcommand{\ev}[2]{\evOp_{#2}(#1)}

\newcommand{\numTwists}{\ell}

\newcommand{\tVec}{{\ve t}}

\newcommand{\etaVec}{{\ve \eta}}

\newcommand{\mybox}{\Box}
\newcommand{\etai}[1]{\eta_{t_{#1}}^{-1}}

\newcommand{\alphaVec}{{\ve \alpha}}
\newcommand{\betaVec}{{\ve \beta}}

\newcommand{\Polys}{\Fqm[x;\sigma]}
\DeclareMathOperator{\rk}{rk}

\newcommand{\Gal}{\mathrm{Gal}}
\newcommand{\dR}{\mathrm{d_R}}
\newcommand{\ZZ}{\mathbb{Z}}

\newcommand{\lambdaVec}{\ve \lambda}
\newcommand{\multitwisted}{$(k, \tVec,\lambdaVec,\etaVec)$-twisted }
\newcommand{\multievpolys}{\mathcal{V}_{k,\tVec,\lambdaVec,\etaVec}}
\newcommand{\Ctwisted}{\mathcal{C}_k(\alphaVec,\tVec,\lambdaVec,\etaVec)}
\newcommand{\Fqsi}[1]{\mathbb{F}_{q^{s_{#1}}}}

\renewcommand{\S}{\mathcal{S}}

\newcommand{\Ann}{\mathcal A}

\newcommand{\mymat}[1]{}

\newcommand{\Fqn}{\mathbb{F}_{q^n}}
\newcommand{\cV}{\mathcal V}

\def\Fq{{\mathbb{F}}_q}

\def\dim{\mathrm{dim}}
\DeclareMathOperator{\modr}{~mod_r}
\DeclareMathOperator{\gcrd}{gcrd}

\newcommand{\malpha}{\Ann_{\alphaVec}}

\newcommand{\eva}[1]{\ev{#1}{\alphaVec}}
\newcommand{\evb}[1]{\ev{#1}{\betaVec}}

\begin{document}

\maketitle

\begin{abstract}
We present a new family of maximum rank distance (MRD) codes.
The new class contains codes that are neither equivalent to a generalised Gabidulin nor to a twisted Gabidulin code, the only two known general constructions of linear MRD codes.
\end{abstract}

\section{Introduction}

Rank-metric codes are sets of matrices, where the distance of two elements is measured with respect to the rank metric, i.e., the rank of their difference.
These codes have found many applications, such as random linear network coding \cite{silva2008rank}, MIMO communication \cite{gabidulin2000space}, cryptography \cite{gabidulin1991ideals}, and distributed storage \cite{silberstein2012error}.
A rank-metric code is called a maximum rank distance (MRD) code if its minimum rank distance is maximal for the given parameters.

The first known class of MRD codes are Gabidulin codes, which were independently introduced in \cite{Delsarte_1978,Gabidulin_TheoryOfCodes_1985,Roth_RankCodes_1991}.
They are evaluation codes of skew polynomials \cite{ore1933theory} that are defined using the Frobenius automorphism $\cdot^q$ of a finite field extension $\Fqm/\Fq$.
Gabidulin codes were generalised in \cite{roth1996tensor,kshevetskiy2005new,augot2013rank} using other automorphisms.

The first families of MRD codes that are not equivalent to a generalised Gabidulin code were independently introduced in \cite{sheekey2015new} (twisted Gabidulin codes) and \cite{otal2016explicit}, where the latter is a special case of the first.
Similar to Gabidulin codes, twisted Gabidulin codes were generalised using different automorphisms in \cite[Remark~9]{sheekey2015new} and \cite{lunardon2015generalized}.
Other constructions, leading to non-linear codes or codes with restricted parameters, can be found in \cite{cossidente2016non,horlemann2015new,otal2017additive}.

Recently, the idea of ``twisting'' was transferred to Reed--Solomon codes, the Hamming metric analogue of Gabidulin codes \cite{beelen2017twisted}.
This transfer resulted in new generalisations of the ``twisting'' idea and tools for the analysis of the new codes.

In this paper, we introduce a new class of rank-metric codes using ideas from \cite{sheekey2015new} and \cite{beelen2017twisted} (cf.~Section~\ref{sec:new_construction}).
A sufficient condition for the new codes to be MRD is derived in Section~\ref{sec:mrd_property}.
In Section~\ref{sec:nonequivalence}, we show that the new family contains MRD codes that are not equivalent to generalised Gabidulin or to the twisted Gabidulin codes in \cite{sheekey2015new}.
Finally, we briefly discuss a possible application to cryptography in Section~\ref{sec:application} and conclude the paper in Section~\ref{sec:conclusion}.

\section{Preliminaries}
\label{sec:preliminaries}

The following are well-known facts on skew polynomials over finite fields, see e.g. \cite{ore1933theory,mcdonald1974finite}.
Let $\Fqm/\Fq$ be a finite extension of a finite field and $\sigma \in \Gal(\Fqm/\Fq)$ be a generator of the Galois group of the extension, i.e., $\sigma = (\cdot)^{q^i}$, where $\gcd(i,m)=1$.
The skew polynomial ring $\Polys$ is the set of formal polynomial expressions $a = \sum_{i} a_i x^i$, where $a_i \in \Fqm$, with ordinary component-wise addition $+$ and the following multiplication rule:
\begin{align*}
\left(\textstyle\sum_{i} a_i x^i \right) \cdot \left(\textstyle\sum_{i} b_i x^i\right) = \textstyle\sum_i \left( \textstyle\sum_{j=0}^{i} a_j \sigma^{j}\left( b_{i-j} \right) \right) x^i.
\end{align*}
We define the degree of $a = \sum_{i} a_i x^i \in \Polys$ by $\deg(a) := \max\{i : a_i \neq 0\}$ (and $\deg 0 := - \infty$) and the (operator) evaluation map \cite{boucher2014linear} by
\begin{align*}
a(\cdot) : \Fqm \to \Fqm, \quad \alpha \mapsto \textstyle\sum_{i} a_i \sigma^i\left(\alpha\right).
\end{align*}
Since $\sigma$ is an automorphism that fixes $\Fq$, the evaluation map of a skew polynomial is an $\Fq$-linear map and its root space $\ker(a)$ is an $\Fq$-linear subspace of $\Fqm$ (seen as an $m$-dimensional vector space over $\Fq$).
Since $\sigma$ is a generator of $\Gal(\Fqm/\Fq)$, we know that $\dim \ker(a) \leq \deg a$.
For any subspace $\S \subseteq \Fqm$, there is a unique monic polynomial $\Ann_\S$, the \emph{annihilator polynomial of $\S$}, of minimal degree whose kernel contains $\S$. Its degree is $\deg \Ann_\S = \dim_{\Fq}(\S)$.

The ring of skew polynomials is left and right Euclidean \cite{ore1933theory}, so the greatest common right divisor $\gcrd$ and the right modulo operator $\modr$ are well-defined.

\begin{definition}\label{def:eval_map}
Let $\mathcal{V}$ be a $k$-dimensional $\Fqm$-linear subspace of $\subset \Polys$.
Let $\alpha_1,\dots,\alpha_n \in \Fqm$ be linearly independent over $\Fq$ and write $\alphaVec = [\alpha_1,\dots,\alpha_n]$.
Then we define the \emph{evaluation map} of $\alphaVec$ on $\mathcal{V}$ by
\begin{align*}
\evOp_{\alphaVec} : \mathcal{V} \to \Fqm^n, \quad
f \mapsto [f(\alpha_1),\dots,f(\alpha_n)].
\end{align*}
We call $\alpha_1,\dots,\alpha_n$ the \emph{evaluation points} of $\evOp_{\alphaVec}$.
\end{definition}

Note that $\evOp_{\vec\alpha}$ is an $\Fqm$-linear map.
Furthermore, $\ker \evOp_{\vec \alpha} = \{ f \Ann_{\vec\alpha} \mid f \in \Polys \} \cap \mathcal V$.
If all elements of $\mathcal V$ have degree less than $n$, $\evOp_{\vec \alpha}$ is invertible.

\begin{definition}
  Given $\alpha_1,\ldots,\alpha_n \in \Fqm$ which are linearly independent over $\Fq$ as well as $k$ with $1 \leq k \leq n$ and some $\sigma \in \Gal(\Fqm/\Fq)$, then the \emph{$[n,k]$ Gabidulin code $\C \subset \Fqm^n$ with defining automorphism $\sigma$ and evaluation points $\vec \alpha = [\alpha_1,\ldots,\alpha_n]$} is the set:
  \[
    \C = \ev{\langle 1, x, \ldots, x^{k-1} \rangle_{\Fqm}\!}{\vec\alpha} \subset \Fqm^n \ ,
  \]
  where $x$ is the indeterminate of $\Polys$.
\end{definition}

From the properties of $\evOp_{\vec\alpha}$, a Gabidulin code is a linear code of dimension $k$.
When $\sigma \neq (\cdot)^q$ then a Gabidulin code with defining automorphism $\sigma$ is sometimes called a ``generalised Gabidulin code''.

For any $a = (a_1,\ldots,a_n) \in \Fqm^n$ we define $\rk(a) := \dim_{\Fq}\langle a_1,\ldots a_n \rangle$.
In other words, $\rk(a)$ is the rank of the $\Fq^{m \times n}$-matrix obtained by expanding each $a_i$ into an $\Fq^m$-vector over any basis of $\Fqm/\Fq$.
We then endow $\Fqm^n$ with the \emph{rank metric over $\Fq$}, for $a, b \in \Fqm^n$ we define $\dR(a,b) := \rk(a-b)$.
A ``rank metric code'' is simply a linear code over $\Fqm$ whose properties are considered wrt.~the rank metric.
A Gabidulin code $\C$ attains the Singleton bound for the rank metric, that is, $\dR(\C) = n - k + 1$.
Codes attaining this bound are called Maximal Rank Distance codes, or MRD.

\section{A New Construction}
\label{sec:new_construction}

\subsection{Idea}
\label{ssec:idea}

Gabidulin codes are MRD since any non-zero skew polynomials $f \in \Polys$ of degree at most $k-1$ has a space of roots of dimension at most $k-1$.
Hence by evaluating $f$ at $n$ linearly independent elements of $\Fqm$, the result spans a subspace of dimension at least $n-k+1$.

The paper \cite{sheekey2015new} considered skew polynomials of the form $f = f' + \eta f_0 x^k$ where $\deg f' \leq k-1$, $f_0$ is the constant coefficient of $f'$, and $\eta \in \Fqm$ is some fixed constant.
He showed that by choosing $\eta$ carefully, any such $f$ will \emph{also} have a space of roots of dimension at most $k-1$, even though $f$ has degree $k$.
Applying the evaluation map to this space of polynomials immediately gives an MRD code which turns out to be inequivalent to a Gabidulin code; these codes were dubbed ``twisted Gabidulin codes''.
\cite{sheekey2015new} mainly discussed $\sigma = (\cdot)^q$, but it is straightforward to use any generator $\sigma \in \Gal(\Fqm/\Fq)$, see \cite[Remark~9]{sheekey2015new} or \cite{lunardon2015generalized}.

In this paper we consider the obvious generalisation of ``twisting'' in other ways than with (a multiple of) $f_0$ at the monomial $x^k$.
One e.g.~immediately thinks of polynomials of the form $f = f' + \eta f_h x^{k-1 + t}$, where $f_h$ is the coefficient to $x^h$ in $f'$ for some $h \in \{0,\ldots,k-1\}$ and $t \in \{1,\ldots,n-k\}$.
``Multiple twists'' is the next natural idea, e.g.~$f = f' + \eta_1 f_{h_1} x^{k-1 + t_1} +  \eta_2 f_{h_2} x^{k-1 + t_2}$.
The difficulty lies in how to argue that the resulting polynomials never have a space of roots of dimension greater than $k-1$, within the space spanned by the evaluation points.
The tools we develop for this are very different from those of \cite{sheekey2015new,lunardon2015generalized} but are completely analogous to those of the recent ``Twisted Reed--Solomon codes'' \cite{beelen2017twisted}.
It turns out that those tools cope effortlessly with an even more general notion of twisting, where the $f_{h_1}, f_{h_2}, \ldots$ are replaced with arbitrary linear combinations of all the $f_0,\ldots,f_{k-1}$, i.e.:
\begin{align*}
  \textstyle
f = \sum_{j=0}^{k-1} f_j x^j + \sum_{i=1}^{\ell} \eta_i \left( \sum_{j=0}^{k-1} \lambda_{i,j} f_j \right) x^{k-1+t_i}.
\end{align*}
for fixed $\eta_1,\ldots,\eta_\ell$ and $\lambda_{i,j}$ for $i=1,\ldots,\ell$ and $j=0,\ldots,k-1$.

This generalization comes at the cost of shorter code lengths $n$ since we will need to restrict the choice of evaluation points.
On the other hand, our codes are very constructive: the restrictions on choosing the evaluation points and the parameters $\eta_i$ and $\lambda_{i,j}$ are simply that they belong to or avoid certain sub-fields of $\Fqm$, and are therefore easily satisfied.

We also call our new codes ``Twisted Gabidulin codes'', and consider the codes of \cite{sheekey2015new,lunardon2015generalized} special twists.

\subsection{Formal definition}
\label{ssec:definition}

Let $n,k,\ell \in \NN$ be such that $k < n \leq m$ and $\ell \leq n-k$. Furthermore, let $\eta_1,\ldots,\eta_\ell \in \Fqm \setminus \{0\}$ as well as $0 < t_1 < t_2 < \ldots < t_\ell < n-k$.
For $i=1,\dots,\ell$, let $\lambda_i: \Fqm^k \rightarrow \Fqm$ be $\Fqm$-linear maps.
Write $\etaVec = [\eta_1,\dots,\eta_\ell]$; $\tVec = [t_1,\dots,t_\ell]$; and $\lambdaVec = [\lambda_1,\dots,\lambda_\ell]$.
We define the set of \emph{\multitwisted skew polynomials} by
\begin{align*}
\multievpolys = \left\{ f = \sum_{i=0}^{k-1} f_i x^i + \sum_{i=1}^{\ell} \eta_i \lambda_i(f_0,\dots,f_{k-1}) x^{k-1+t_i} : f_i \in \Fqm \right\}.
\end{align*}

\begin{definition}
Let $n,k,\ell,\tVec,\lambdaVec,\etaVec$ be as above and $\alphaVec := [\alpha_1,\dots,\alpha_n] \in \Fqm$ be such that $\alpha_1,\dots,\alpha_n$ are linearly independent over $\Fq$.
The corresponding ($\ell$-)\emph{twisted Gabidulin code} is defined by
\begin{align*}
\Ctwisted := \ev{\multievpolys}{\alphaVec} \subseteq \Fqm^n \ .
\end{align*}
\end{definition}

Since $\multievpolys$ is an $\Fqm$-linear subspace of $\Polys_{<n}$ and $\ev{\cdot}{\alphaVec}$ is an injective linear map on $\Polys_{<n}$, a twisted Gabidulin code is an $\Fqm$-linear code of length $n$ and dimension $k$.

Note that given $\lambda_i$, there are unique constants $\lambda_{i,0},\ldots,\lambda_{i,k-1} \in \Fqm$, not all zero, and such that $\lambda_i(f_0,\ldots,f_{k-1}) = \lambda_{i,0} f_0 + \ldots + \lambda_{i,k-1} f_{k-1}$. We will call the $\lambda_{i,j}$ the \emph{coefficients} of $\lambda_i$.
In the sequel, whenever we introduce $\vec \lambda$, then we take $\lambda_{i,j}$ to be the coefficients of the entries of $\vec \lambda$.

\section{Twisted Gabidulin Codes that are MRD}
\label{sec:mrd_property}

Not all twisted Gabidulin codes as defined in the preceding section will be MRD.
In this section, we first give a linear-algebraic condition for when this will be the case, and we then use this to describe an explicit family of twisted Gabidulin codes which are MRD.

\begin{lemma}
\label{lem:MRD_lemma}
Let $n,k,\tVec,\lambdaVec,\etaVec$ be chosen as in \cref{ssec:definition}.
For any $\S \subset \langle \alpha_1,\dots,\alpha_n \rangle_{\Fq}$ with $\dim_{\Fq} \S = k$, consider the homogeneous, linear system of equations in the $g_0,\ldots,g_{t_\ell-1}$:
\begin{equation}
\label{eq:MRD_lemma_system}
  \sum_{j=0}^{t_\ell - 1} g_j T^{(\S)}_{i, j} = 0  \quad \textrm{for } i = k, \ldots, k - 1 + t_\ell \ ,
\end{equation}
where
\begin{align}
  T^{(\S)}_{i,j} &= \left\{
  \begin{array}{ll}
     \eta_\kappa^{-1} \sigma^{j}(a_{i-j}) - \sum_{\mu=0}^{k-1} \lambda_{\kappa,\mu} \sigma^{j}(a_{\mu-j})  & \textrm{ if } i = k-1 + t_\kappa \textrm{ for } \kappa \in \{ 1,\ldots,\ell \}
    \\
    \sigma^j(a_{i-j}) & \textrm{otherwise}
  \end{array}
  \right. \ , \notag
\end{align}
and $\sum\limits_{i=0}^{k} a_i x^i := \Ann_\S(x)$ and $a_i := 0$ for $i < 0$ or $i \geq k$.
The twisted Gabidulin code $\Ctwisted$ is MRD if and only if there is no choice of $\S$ admitting a non-zero solution to the linear system.
\end{lemma}

\begin{proof}
Let $f \in \multievpolys$ be a polynomial whose root space intersects with $\langle \alpha_1,\dots,\alpha_n \rangle$ in at least $k$ dimensions, i.e.,
\begin{align*}
\dim \left(\ker(f) \cap \langle \alpha_1,\dots,\alpha_n \rangle \right) \geq k.
\end{align*}
This means that there is a $k$-dimensional subspace $\S \subseteq \langle \alpha_1,\dots,\alpha_n \rangle$ such that the annihilator polynomial $\Ann_{\S}$ divides $f$ from the right, i.e. $f = g \cdot \Ann_\S$ for some $g \in \Polys$.
Write $g = \sum_{i=0}^{t_\ell-1} g_i x^i$ and $\Ann_\S = \sum_{i=0}^{k} a_i x^i$, as well as $a_i := 0$ for $i < 0$ or $i \geq k$.
Then the $i$-th coefficient of $f$ is given by
\begin{align*}
f_i = \sum\limits_{j=0}^{t_\ell-1} g_j \sigma^{j}\left( a_{i-j} \right) \ .
\end{align*}
On the other hand, for $i \geq k$, then $f_i$ also satisfies
\begin{align*}
f_{i} =
\begin{cases}
\eta_\kappa \lambda_\kappa(f_0,\dots,f_{k-1}) & \textrm{if } i=k-1+t_\kappa \text{ for } \kappa \in \{1,\dots,\ell\}, \\
0 &  \textrm{otherwise} \ .
\end{cases}
\end{align*}
For the case $i=k-1+t_\kappa$, combining and rewriting yields
\begin{align*}
0 &= \sum_{j=0}^{t_\ell-1} g_j \cdot \left( \eta_\kappa^{-1} \sigma^{j}\left( a_{i-j}\right) - \sum_{\mu=0}^{k-1} \lambda_{\kappa,\mu} \sigma^{j}\left( a_{\mu-j} \right) \right).
\end{align*}
Thus, for a given $\S$, a non-zero solution of the linear system \eqref{eq:MRD_lemma_system} corresponds to a non-zero polynomial $f = g \cdot \Ann_{\S} \in \multievpolys$ with
\begin{align*}
\rk(\ev{f}{\alphaVec}) = n-\dim \left(\ker(f) \cap \langle \alpha_1,\dots,\alpha_n \rangle \right) \leq n-k.
\end{align*}
The code $\Ctwisted$ is MRD if and only if all non-zero $f \in \multievpolys$ result in codewords of rank $\geq n-k+1$, which is true if and only if the system \eqref{eq:MRD_lemma_system} has no non-zero solution for any subspace $\S$. \qed
\end{proof}

For a given subspace $\S$, the system \eqref{eq:MRD_lemma_system} is of the form
\setcounter{MaxMatrixCols}{20}
\begin{align}\tiny
\underset{=: \, \B_\S}{\underbrace{\begin{bmatrix}
\etai{\numTwists} + \mybox& \mybox& \dots & \mybox & \mybox & \mybox & \dots & \mybox & \mybox & \dots\\
\mybox & 1 & \mymat{0} & \mymat{\dots} & \mymat{0} & \mymat{0} & \mymat{0} & \mymat{\dots} & \mymat{0} &  \mymat{0} & \mymat{\dots} \\
\vdots & \vdots & \ddots & \mymat{\vdots} & \mymat{\vdots} & \mymat{\vdots} & \mymat{\ddots} & \mymat{\vdots} & \mymat{\vdots} & \mymat{\dots} \\
\mybox & \mybox & \dots & 1 & \mymat{0} & \mymat{0} & \mymat{\dots} & \mymat{0} & \mymat{0} & \mymat{\dots} \\
\etai{\numTwists-1} \mybox + \mybox &\etai{\numTwists-1} \mybox + \mybox & \dots & \etai{\numTwists-1} \mybox + \mybox & \etai{\numTwists-1} + \mybox & \mybox & \dots & \mybox & \mybox & \dots \\
\mybox & \mybox & \dots & \mybox & \mybox & 1 & \mymat{\dots} & \mymat{0} & \mymat{0} \mymat{\dots}  \\
\vdots & \vdots & \ddots & \vdots & \vdots & \vdots & \ddots & \mymat{\vdots} & \mymat{\vdots}  \mymat{\dots}\\
\vdots & \vdots & \ddots & \vdots & \vdots & \vdots & \ddots & \mymat{\vdots} & \mymat{\vdots} \mymat{\dots} \\
\mybox & \mybox & \dots & \mybox & \mybox & \mybox & \dots & 1 & \mymat{0}  & \mymat{\dots}  \\
\etai{1} \mybox + \mybox & \etai{1} \mybox + \mybox & \dots & \etai{1} \mybox + \mybox & \etai{1} \mybox + \mybox & \etai{1} \mybox + \mybox & \dots & \etai{1} \mybox + \mybox & \etai{1} + \mybox & \dots \\
\vdots & \vdots & \ddots & \vdots & \vdots & \vdots & \ddots & \vdots & \vdots & \ddots
\end{bmatrix}}}
\begin{bmatrix}
g_{t_{\numTwists}} \\
g_{t_{\numTwists}-1} \\
\vdots \\
g_{1} \\
g_{0}
\end{bmatrix}
=
\ve{0}, \label{eq:B_S_matrix}
\end{align}
where the boxes $\mybox$ represent elements in the $\Fq$-span of the $\lambda_{i,j}$, $\alpha_1,\dots,\alpha_n$, and their $q$-powers.
The diagonal elements are either $1$ (if the row corresponds to an index $i$ with $i \neq k-1+t_\kappa$ for all $\kappa$) or $\etai{\kappa} + \mybox$ (if $i \neq k-1+t_\kappa$) due to $\sigma^{i-k}(a_{i - (i-k)}) = \sigma^{i-k}(a_{k}) = 1$ for all $i=k,k+1,\dots,k-1+t_\numTwists$ (note that the diagonal elements are the $T_{i,j}^{(\S)}$ of Equation~\eqref{eq:MRD_lemma_system} with $j=i-k$).
Also, all elements above the diagonal do not depend on the $\etai{\kappa}$ since $a_{i-j} = 0$ for all $i>j+k$.

Using \cref{lem:MRD_lemma}, we can give the following sufficient condition for a twisted Gabidulin code to be MRD.

\begin{theorem}\label{thm:MRD_sufficient_condition}
Let $s_0,\dots,s_\ell \in \NN$ such that $\Fq \subseteq \Fqsi{0} \subsetneq \Fqsi{1} \subsetneq \dots \subsetneq \Fqsi{\ell} = \Fq$ is a chain of subfields.
Let $k < n \leq s_0$ and $\alpha_1,\dots,\alpha_n \in \Fqsi{0}$ be linearly independent over $\Fq$, and let $\tVec$, $\lambdaVec$, and $\etaVec$ be chosen as in \cref{ssec:definition} with the additional requirements $\eta_i \in \Fqsi{i} \setminus \Fqsi{i-1}$ and $\lambda_{i,j} \in \Fqsi{0}$ for all $i,j$.
Then, the twisted Gabidulin code $\Ctwisted$ is MRD.
\end{theorem}

\begin{proof}
We prove the claim using \cref{lem:MRD_lemma}.
Let $\S \subseteq \langle \alpha_1,\dots,\alpha_n \rangle$ be a $k$-dimensional subspace.
We show that the system \eqref{eq:MRD_lemma_system} has no non-zero solution.
Since $\lambda_{i,j} \in \Fqsi{0}$, $a_i \in \langle \alpha_1,\dots,\alpha_n \rangle \subseteq \Fqsi{0}$, $\sigma(\Fqsi{0}) = \Fqsi{0}$, the boxes $\mybox$ of the system's matrix $\B_\S$ as in \eqref{eq:B_S_matrix} represent elements from $\Fqsi{0}$.

We now consider the $\etai{\kappa}$'s to be indeterminates.
This means that $\det(\B_\S) \in \Fqsi{0}[\etai{1},\dots,\etai{\numTwists}]$ is a multivariate polynomial, where each indeterminate $\etai{\kappa}$ appears at most of degree $1$ in each monomial.
Let $\B_\S^{(\mu)}$ be the $(\mu \times \mu)$-bottom-right submatrix of $\B_\S$. We distinguish two cases:
\begin{enumerate}[label=(\roman*)]
\item If $\mu \neq t_\kappa$ for all $\kappa$, then the first row of $\B_\S^{(\mu)}$ is of the form $[1, 0, \dots, 0]$ and by Laplace's rule we get
\begin{align*}
\det\left(\B_\S^{(\mu)}\right) = \det\left(\B_\S^{(\mu-1)}\right).
\end{align*}
\item If $\mu = t_\kappa$ for some $\kappa$, then $\B_\S^{(\mu)}$ contains only $\etai{1},\dots,\etai{\kappa-1}$ in its rows $2$ to $\mu$ and since the first row is of the form $[\etai{\kappa}+\mybox, \mybox, \dots, \mybox]$, the determinant fulfills
\begin{align*}
\det\left(\B_\S^{(\mu)}\right) = (\etai{\kappa} + T_\mu) \cdot \det\left(\B_\S^{(\mu-1)}\right) + U_\mu \in \Fqsi{0}[\etai{1},\dots,\etai{\kappa}],
\end{align*}
where $T_\mu \in \Fqsi{0}$ and $U_\mu \in \Fqsi{0}[\etai{1},\dots,\etai{\kappa-1}]$.
\end{enumerate}
Combined, we get $\det(\B_\S^{(t_\kappa)}) = (\etai{\kappa} + T_{t_\kappa}) \det(\B_\S^{(t_{\kappa-1})}) + U_{t_\kappa} \in \Fqsi{0}[\etai{1},\dots,\etai{\kappa}]$, where $\det(\B_\S^{(t_1)}) = \etai{1}$, and by recursively substituting $\eta_\kappa \in \Fqsi{\kappa} \setminus \Fqsi{\kappa-1}$ for $\kappa=1,\dots,\numTwists$,  we obtain
\begin{align*}
\det\left(\B_\S^{(t_\kappa)}\right) \in \Fqsi{\kappa} \setminus \{0\},
\end{align*}
since $\etai{\kappa} \in \Fqsi{\kappa} \setminus \Fqsi{\kappa-i}$, $\det(\B_\S^{(t_{\kappa-1})}) \in \Fqsi{\kappa-1} \setminus \{0\}$, and $U_{t_\kappa} \in \Fqsi{\kappa-1}$.
Hence, also $\det( \B_\S ) = \det(\B_\S^{(t_\numTwists)}) \neq 0$ and System \eqref{eq:MRD_lemma_system} has only the zero solution. \qed
\end{proof}

Theorem~\ref{thm:MRD_sufficient_condition} provides a tool to systematically construct MRD twisted Gabidulin codes.
For some $m$, we can obtain codes of length up to $2^{-\ell}m$ in this way:

\begin{corollary}
Let $\ell \in \ZZ_{>0}$ and $2^\ell \mid m$. Then there is an $\ell$-twisted MRD code of length $n=2^{-\ell}m$ over $\Fqm$.
\end{corollary}

\section{Non-Equivalence to Other MRD Codes}
\label{sec:nonequivalence}

In this section, we show that the new family of twisted Gabidulin codes contains codes that are neither equivalent to a generalised Gabidulin nor to the twisted Gabidulin codes constructed in \cite{sheekey2015new}.
We use the following notion for equivalence of linear rank-metric codes.

\begin{definition}[\!\cite{morrison2014equivalence}]\label{def:equivalence}
Two linear rank-metric codes $\Code,\Code' \subseteq \Fqm^n$ are (semi-linearly) equivalent if there are $\lambda \in \Fqm^\ast$, $\A \in \mathrm{GL}_n(q)$, and $\sigma \in \Gal(\Fqm/\Fq)$ such that
\begin{align*}
\Code' = \sigma(\lambda \Code) \A,
\end{align*}
where $\sigma(\lambda \Code) \A := \left\{ \left[ \sigma(\lambda c_1), \dots, \sigma(\lambda c_n)  \right] \cdot \A :  \left[c_1, \dots,c_n \right] \in \Code \right\}$.
\end{definition}

Since for twisted Gabidulin codes, we can only guarantee them to be MRD for $n<m$, we cannot directly rely on the tools developed in \cite{sheekey2015new} for proving the inequivalence to Gabidulin codes. We first need to state two lemmas.

\begin{lemma}\label{lem:fg_alpha_beta}
Suppose $\alphaVec= [\alpha_1,\cdots,\alpha_n]$ is a list of elements of $\Fqm$, linearly independent over $\Fq$, and $f,g \in \Fqn[x,\sigma]$, with $\deg(f),\deg(g)<m$. Then 
\[
\ev{f}{\alphaVec}= \ev{g}{\alphaVec} \Leftrightarrow f-g \equiv 0 \modr \Ann_{\alphaVec}
\]
\end{lemma}

\begin{lemma}
\label{lem:equiv}
Suppose $\cV,\cV'$ are two $\Fqm$-subspaces of $\Fqm[x,\sigma]$, and $\alphaVec,\betaVec$ two lists of $n$ elements of $\Fqm$, each linearly independent over $\Fq$. Then $\eva{\cV}$ is semilinearly equivalent to $\evb{\cV'}$ if and only if there exist elements $\psi,\phi\in \Fqm[x,\sigma] $ such that $\psi$ is a monomial, $\gcrd(\phi,x^m-1)=1$, and
\[
\{f \modr \malpha :f \in \cV\} = \{\psi g\phi \modr \malpha:g\in \cV'\}.
\]
Furthermore, if $\eva{\cV}$ is equivalent to $\evb{\cV'}$ then there exist $\psi,\phi\in \Fqm[x,\sigma] $ such that $\psi$ is monomial, $\deg(\phi')<n$, $\gcrd(\phi',\malpha)=1$, and 
\[
\{f \modr \malpha :f \in \cV\} = \{\psi g\phi' \modr \malpha:g \in \cV'\}.
\]
\end{lemma}

\begin{proof}
Suppose $\eva{\cV}=\sigma^i(\lambda\evb{\cV'})\mathbf{A}$.  Let $\psi=\lambda^{\sigma^{\ell}} x^{\ell}$, and let $\phi$ be such that $\phi(\alpha_i)=\sum_j A_{ji}\beta_j$, where the $A_{ji}$'s are the entries of $\mathbf{A}$. As $\mathbf{A}$ is invertible, and as $\alphaVec$ and $\betaVec$ are linearly independent, $\phi$ can be chosen so that it does not evaluate to zero on any element of $\Fqm^{*}$; in other words, $\gcrd(\phi,x^m-1)=1$.
Now 
\begin{align*}
\sigma^\ell(\lambda\evb{g})\mathbf{A} &= \left(\lambda^{\sigma^\ell}(x^\ell g)\left(\sum_j A_{j1}\beta_j\right),\ldots,\lambda^{\sigma^\ell}(x^\ell g)\left(\sum_j A_{jn}\beta_n\right)\right) \\
&=(\lambda^{\sigma^\ell}(x^\ell g)(\phi(\alpha_1)),\ldots,\lambda^{\sigma^\ell}(x^\ell g)(\phi(\alpha_1)) = \ev{\psi g\phi}{\alpha},
\end{align*}
and so 
\begin{align*}
\{\eva{f}:f \in \cV\} = \{\evb{g}:g\in \cV'\} = \{ \eva{\psi g\phi}:g\in \cV'\},
\end{align*}
so taking this together with Lemma~\ref{lem:fg_alpha_beta} gives the first result.

For the second part, we take $\phi'$ to be the remainder of $\phi$ on right division by $\malpha$, i.e. $\phi = a\malpha +\phi'$, with $\deg(\phi')<\deg(\malpha)=n$. Then the result follows immediately. \qed
\end{proof}

In the following, we consider the special case of twisted Gabidulin codes, given by evaluation polynomials of the form
\begin{align}
\cV_{k,t,\eta} := \left\{\sum_{i=0}^{k-1} f_ix^i + \eta f_0 x^{k-1+t}:f_i\in \Fqm\right\}. \label{eq:VsimpleTwists}
\end{align}
Note that these are contained in the codes constructed in \cite{sheekey2015new} precisely when $t=1$, and with Gabidulin codes precisely when $\eta=0$. The following theorem proves that any twisted Gabidulin code $\ev{\cV_{k,t,\eta}}{\alphaVec}$ with $1< t<s-1$ and $\eta \neq 0$ is not equivalent to one of these code classes. Since generalised Gabidulin codes are equivalent to Gabidulin codes (cf.~\cite{horlemann2015new}), such twisted Gabidulin codes are also not equivalent to one of them.

\begin{theorem}
Let $\alphaVec$ be an $\Fq$-basis for $\Fqs\leq \Fqm$. Let $1< t<s-1$ and $\eta \neq 0$. Then the code $\eva{\cV_{k,t,\eta}}$ is not equivalent to $\evb{\cV_{k,1,\eta'}}$ for any $\betaVec,\eta'$.
\end{theorem}

\begin{proof}
We have that $\malpha = x^s-1$. As $x^i\in \cV_{k,1,\eta'}$ for $i \in \{1,\ldots,k-1\}$, by Lemma~\ref{lem:equiv} we must have 
\begin{align*}
\psi x^i \phi \modr (x^s-1) \in \cV_{k,t,\eta}
\end{align*}
Let $\psi = x^\ell$ and $\phi = \sum_{j=0}^{s-1} \phi_j x^j$. Note that $x^\ell \cV_{k,t,\eta} x^{s-\ell} = \cV_{k,t,\eta}$, and  $x^\ell \cV_{k,1,\eta'} x^{s-\ell} = \cV_{k,1,\eta'}$, and so without loss of generality we may assume that $\ell=0$.
Hence we have
\begin{align*}
\sum_j \phi_j^{\sigma^{i}}x^{i+j \mod s} \in \cV_{k,t,\eta}
\end{align*}
for all $i \in \{1,\ldots,k-1\}$. As the coefficient of $x^r$ is zero for every element of $\cV_{k,t,\eta}$ for $r\notin \{0,\ldots,k-1\}\cup \{k+t-1\}$, we must have that $\phi_j=0$ whenever $i+j \notin \{0,\ldots,k-1\}\cup \{k+t-1\}$. Hence we get that $\phi = \phi_{0}+\phi_{s-1}x^{s-1}$. 

If $\eta'=0$, then as $1\in \cV_{k,1,\eta'}$ we get that $\phi_{0}+\phi_{s-1}x^{s-1}\in \cV_{k,t,\eta}$, which implies $\phi=0$, a contradiction. Suppose now that $\eta' \ne 0$. Then as $1+\eta' x^{k}\in \cV_{k,1,\eta'}$, we get that 
\begin{align*}
(1+\eta' x^{k})(\phi_{0}+\phi_{s-1}x^{s-1}) \modr (x^s-1)\in \cV_{k,t,\eta}.
\end{align*}
But this is equal to 
\begin{align*}
\phi_0+\phi_{s-1}x^{s-1}+\eta' \phi_0^{\sigma^k}x^k+\eta'\phi_{s-1}x^{k-1},
\end{align*}
and so as $t>1$ the coefficient of $x^k$ must be zeros, implying $\phi_0=0$, and as $s-1>k+t$ we must have $\phi_{s-1}=0$, implying $\phi=0$, a contradiction. \qed
\end{proof}

\section{Possible Application}
\label{sec:application}

The rank-metric variant GPT \cite{gabidulin1991ideals} of the McEliece public-key cryptosystem \cite{mceliece1978public} was a potential candidate to reduce the size of the public key until structural attacks on the system were found \cite{gibson1995severely,gibson1996security}. Since then, many variants of GPT have been proposed and broken, see~\cite{overbeck2008structural} for an overview.

Structural attacks on a variant of the McEliece cryptosystem based on the twisted Gabidulin codes in \cite{sheekey2015new} can be efficiently performed since any such code of dimension $k$ is a subcode of a Gabidulin code of dimension $k+1$. Such an immediate attack is not possible for the new twisted Gabidulin codes since such a code is only subcode of a $(k-1+t_\ell)$-dimensional Gabidulin code, which does not help the attacker if $t_\ell$ is sufficiently large.

Therefore, twisted Gabidulin codes should be thorougly analysed for their suitability in the McEliece cryptosystem.
An immediate necessity is the existence of an efficient decoding algorithm, but also possibilities for structural attacks should be investigated.

\section{Conclusion}
\label{sec:conclusion}

We introduced a new constructive class of rank-metric codes, twisted Gabidulin codes, that contain codes inequivalent to existing classes, such as Gabidulin or the twisted Gabidulin codes in \cite{sheekey2015new}.
It was proved that MRD twisted Gabidulin codes exist for $n$ up to $2^{-\ell}m$.
Similar to the results on twisted Reed--Solomon codes \cite{beelen2017twisted}, longer codes might be possible.

\bibliographystyle{splncs03}
\bibliography{main}

\end{document}